\documentclass[10pt,journal,compsoc]{IEEEtran}
\usepackage{caption}
\ifCLASSOPTIONcompsoc
  \usepackage[nocompress]{cite}
\else
  \usepackage{cite}
\fi
\usepackage{titlesec}

\usepackage{amsmath}
\usepackage{amsthm}
\usepackage{amssymb}
\DeclareMathOperator*{\argmax}{arg\,max}

\usepackage{graphicx}
\usepackage{textcomp}
\usepackage{cleveref}
\usepackage{xcolor}
\usepackage{booktabs}
\usepackage{multirow}
\usepackage{makecell}
\usepackage{dsfont}
\usepackage{algorithm}
\usepackage[noend]{algpseudocode}
\makeatletter
\def\BState{\State\hskip-\ALG@thistlm}
\makeatother
\newcommand{\supp}{\operatorname{supp}}

\newcommand{\E}{\mathbb{E}}

\newcommand{\tp}{\mathsf{T}}

\newcommand{\R}{\mathbb{R}}

\newtheorem{theorem}{Theorem}
\newtheorem{corollary}{Corollary}

\newtheorem{proposition}{Proposition}
\newtheorem{definition}{Definition}

\newtheorem{assumption}{Assumption}

\usepackage{matlab-prettifier}
\usepackage{color}
\usepackage{mathrsfs}
\usepackage[shortlabels]{enumitem}
\makeatletter
\def\BState{\State\hskip-\ALG@thistlm}
\makeatother
\def\BibTeX{{\rm B\kern-.05em{\sc i\kern-.025em b}\kern-.08em
    T\kern-.1667em\lower.7ex\hbox{E}\kern-.125emX}}
\usepackage{subcaption}

\begin{document}
%
\title{Transparent Tagging for Strategic Social Nudges on User-Generated Misinformation}
%
%
%
%

\author{Ya-Ting~Yang, Tao~Li, and Quanyan~Zhu        
\IEEEcompsocitemizethanks{\IEEEcompsocthanksitem The Authors are with the Department of Electrical and Computer Engineering, New York University, Brooklyn, NY, 11201, USA; E-mail: {\tt\small \{yy4348, tl2636, qz494\}@nyu.edu}. \protect\\
\IEEEcompsocthanksitem Y-T. Yang and T. Li have contributed equally. Correspondence should be addressed to T. Li (\texttt{tl2636@nyu.edu}).}}
\IEEEtitleabstractindextext{%
\begin{abstract}
Social network platforms (SNP) rely heavily on user-generated content to attract users, yet they have limited control over content provision, which leads to misinformation. As countermeasures, SNPs have implemented policies to notify users by tagging the content and influencing users' responses to the tagged content. The population-level response creates a social nudge to the content provider that encourages it to supply more authentic content. Yet, when designing tags to leverage social nudges, SNP must be cautious about misdetection, which impairs its ability to create social nudges. We establish a Bayesian persuaded branching process to study SNP's tagging policy design under misdetection. Misinformation circulation is modeled by a multi-type branching process, where users are persuaded through tags to give positive/negative comments that influence misinformation spread. When translated into posterior belief space, the SNP's problem is reduced to an equality-constrained optimization, the optimal condition of which is given by the Lagrangian characterization. The key finding is that SNP's optimal policy is transparent tagging, albeit misdetection, which nudges the provider not to generate misinformation.  
\end{abstract}

}

\maketitle

\IEEEdisplaynontitleabstractindextext

%
\IEEEpeerreviewmaketitle

\section{Introduction}
Social network platforms (SNP), such as X and TikTok, where users create and consume content, play an increasingly important role in society. These platforms rely heavily on user-generated content (UGC) to engage and retain users to maintain high-level daily activity. Since users who generate original content(``content providers'') are not paid workers, platforms have limited control over the UGC, including misinformation.  

User-generated misinformation has become a growing concern on SNPs, as false information can spread rapidly and have significant consequences \cite{Zhao2020-zu}. For instance, false stories about candidates were shared widely through SNPs during the 2016 US presidential election; misinformation about the virus, mask-wearing policies, and vaccine concerns spread through social networks during the COVID-19 pandemic. To address this issue, SNPs have implemented policies such as labeling, tagging, or notifying to alert users to potentially false or misleading information \cite{twitter, fb_transparency}. 

\begin{figure}
    \centering
    \includegraphics[width=3.6in]{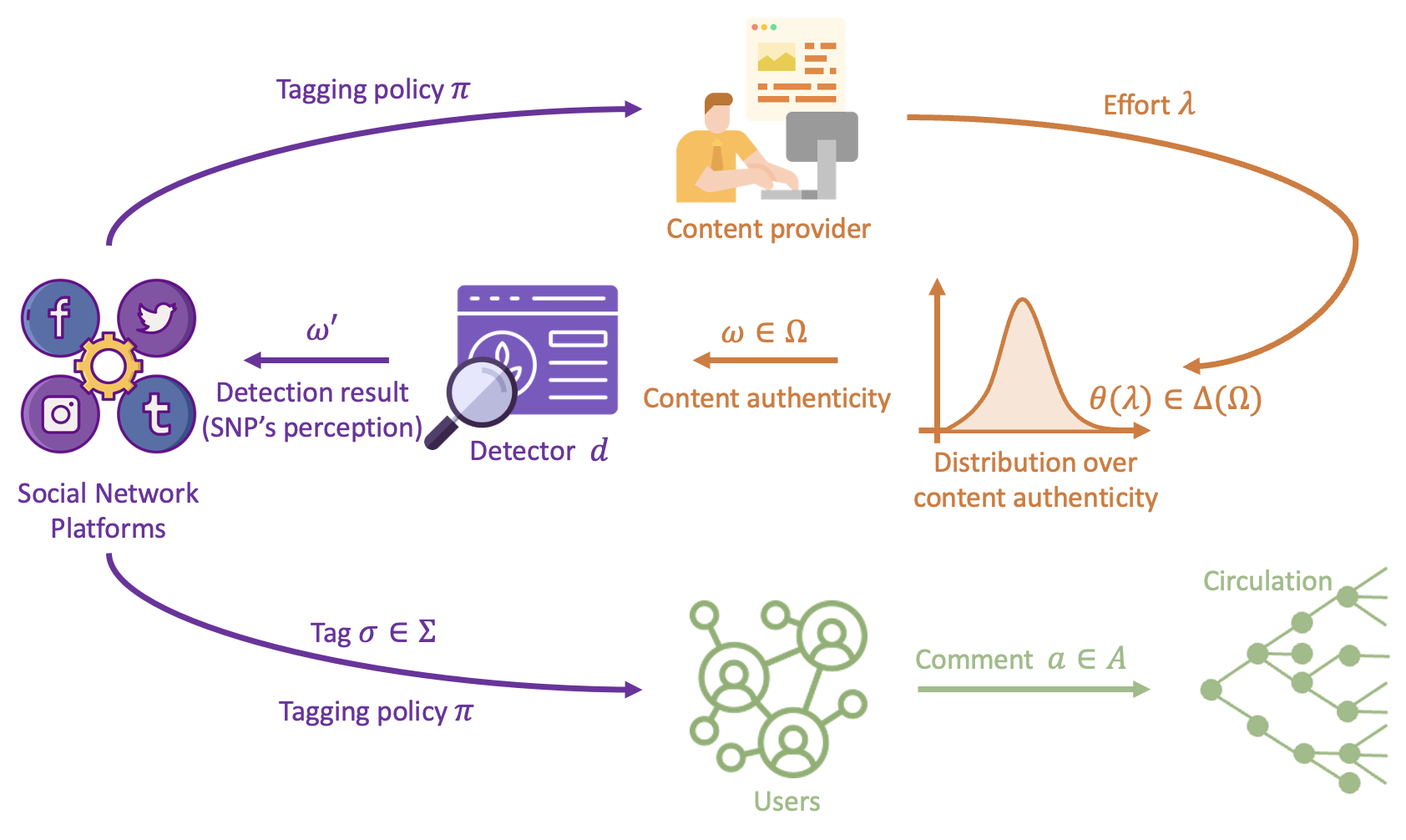}
    \caption{An illustration of the proposed persuasion model, where the misinformation distribution $\theta(\lambda)$ is affected by the content provider and remains unknown to the user. The SNP's misdetection of the underlying content is modeled by $d$.}
    \label{fig:per-mis}
\vspace{-3mm}
\end{figure}

Previous studies have shown that these policies effectively (to some extent) curb the spread of misinformation \cite{ Platform_intervention}. One of the key reasons is that these platforms feature intensive social interactions among users, which can be leveraged to create social nudges in stimulating UGC supply \cite{zeng23social-nudges}. For example, a post tagged as misleading will inflict users' negative comments. After circulation on social networks, the population response to the post creates pressure on the content provider, discouraging it from generating misinformation.   

This work proposes a persuasion game model to provide theoretical underpinnings for the SNP's tagging design, aiming to harness the power of social nudges to reduce user-generated misinformation. As illustrated in \Cref{fig:per-mis}, the strategic interactions among the SNP, the content provider, and the user unfold as below. The SNP designs a tagging policy whose realized tags indicate the content authenticity of an arbitrary post returned by a detection device. Of particular note is that the detection device, usually empowered by artificial intelligence methods \cite{hakak2021ensemble,islam2020deep,10631663}, is often imperfect and may misclassify the post's authenticity. Such a tagging policy does not directly control the content provider or user but influences others' behaviors through information provision. Hence, this tagging policy is referred to as the information structure \cite{tao_info}. Fully aware of this policy, the content provider exerts a private effort (unobservable to the SNP or user) in creating the content, assuming that the more effort exerted, the more authentic the content is. Finally, the user observes the tagging policy and the realized tags and then decides on their views and comments that influence the online circulation modeled by a multi-type branching process.

The proposed model differs from the seminal Bayesian persuasion game \cite{kamenica11BP} in that the user cannot directly observe the prior distribution. Consequently, the user must form a conjecture about the content provider's behavior to update their beliefs. This conjecture must be consistent with the provider's equilibrium behavior, which leads to the concept of perfect Bayesian equilibrium (PBE) as the natural solution concept for our game. 
One prior work \cite{yang2023designing} addressed a special case where there was no detection error, allowing the SNP to identify misinformation in posts perfectly. However, in practical scenarios, detection errors are inevitable. In this work, the SNP's design problem considers such misdetection, which leads to the SNP's misperception of the game state that impairs the tagging policy's credibility and effectiveness in fostering social nudges.

Our key finding is that transparent tagging, where the SNP honestly discloses the detection outcome to the content provider and user, is most effective in combating misinformation generation and circulation. Although the SNP may not have direct control over content generation, it can nudge user perceptions through tagging. The collective behaviors of users, under these perceptions, determine the content provider's reputation, effectively making users the SNP's proxy in terms of incentive provision, encouraging the provider to exert the best effort in reducing misinformation generation. \textbf{Our contributions} are summarized below. 
\begin{itemize}
    \item We propose a three-player Bayesian persuasion game that studies the SNP’s tagging policy under the presence of misdetection and the content provider’s intention to uphold its reputation, with misinformation circulation among users modeled as a multi-type branching process.
    \item We identify players' strategies under perfect Bayesian equilibrium by transforming the problem into the posterior belief space, reducing it to an equality-constrained convex optimization problem.
    \item We characterize the optimal conditions using a Lagrangian approach, demonstrating that the SNP's optimal policy is transparent tagging despite detection errors, incentivizing the content provider to exert maximum implementable effort.
\end{itemize}

\section{Literature Review}

Existing research on misinformation mainly explores scenarios involving a finite set of players (users), typically modeled as nodes in a graph, with the reliability of articles, news, and other content drawn from a ``known'' distribution \cite{model_online_mis, enga_mis}. This line of work often explores how misinformation spreads through different networks and the roles different factors play in circulation. For example, \cite{model_online_mis} introduces a model that analyzes the online sharing behavior of fully Bayesian users when faced with potential misinformation. This study highlights the significant impact of network structure on misinformation propagation, demonstrating that platforms designed to maximize user engagement may inadvertently facilitate the spread of false information. \cite{enga_mis} considers two common objectives for platforms: maximizing user engagement or minimizing the spread of misinformation. By analyzing different strategies, the research provides insights into how platforms can either contribute to or mitigate the dissemination of false content, depending on their underlying goals. Additionally, \cite{sasaki2024misinformation} focuses on how content moderation policies can be designed to enable dominant platforms to enforce regulations without losing users or news sources to competing platforms.

In contrast, our approach considers the population-wide effects of misinformation circulation \cite{papanastasiou2020fake} to examine broader social dynamics and impacts. Specifically, we analyze the proportion of individuals receiving negative comments among all receivers using branching processes, which is shown to closely align with the statistical characteristics of information cascades observed in real-world social media platforms, such as those on Twitter \cite{BP_cascade_twitter}. Besides, results from branching processes have also been utilized in identifying key determinants behind the spread of misinformation \cite{spreading_mis}. Rather than analyzing misinformation circulation through branching processes \cite{BPBP_tagging}, our approach takes a proactive stance by aiming to prevent misinformation from being created in the first place. We extend the classical Bayesian persuasion framework \cite{kamenica11BP} by introducing a third player—the content provider. This addition shifts the focus from merely understanding how misinformation spreads or mitigating misinformation \cite{yang2024prada} to actively controlling its generation. In our model, the SNP aims to curb misinformation spread by incentivizing content providers to produce authentic and truthful content.

In practice, verifying whether a post contains misinformation involves costs and potential errors during the platform’s detection process \cite{aimeur2023fake}. For instance, human-based detection methods, such as crowdsourcing \cite{micallef2020role}, audit \cite{yang2025herd}, and fact-checking \cite{chung2021learn}, often depend on human (expertise) to verify content and are not only time-intensive but the effectiveness of fact-checking initiatives remain questionable \cite{andersen2020communicative}. In contrast, AI-based methods, including classical machine learning \cite{hakak2021ensemble}, deep learning \cite{islam2020deep}, as well as foundation models \cite{10631663, xie2024learning}, provide faster detection but require significant computational resources and still face inevitable detection errors. In this work, we address these limitations by incorporating the detection errors, whether from the detection algorithms or resource limitations, into the design of the platform's tagging policy, enhancing the previous framework \cite{yang2023designing} by considering the platform's real-world challenges.

\section{Online Misinformation Circulation: A Bayesian Persuasion Modeling}

This section introduces a three-player persuasion game that models the interactions between an SNP, content providers, and users. Misinformation circulation on the SNP typically involves many content providers and users. However, to simplify our analysis, we focus on a representative content provider and a homogeneous population of users with identical utilities. For strategic reasoning within the persuasion game, we refer to a representative user as ``the user'' since all users share the same interests. Conversely, when discussing population-level misinformation dissemination using branching processes, we refer to the collective as ``users''.

\subsection{The Bayesian Persuaded Branching Processes Model}
In this persuasion game, the SNP (sender) designs a tagging policy (signaling scheme) about an unknown state that reflects the authenticity of the content of the post (state). The content provider (agent), fully aware of the tagging policy, exerts a private effort in creating the content, which is unobservable to both the SNP and the user (receiver). As the content provider represents a population of providers, the level of effort put into determining the truth influences the content's authenticity, with lower effort leading to more misinformation prevailing over SNPs. Finally, the user takes action by commenting on the post and sharing it with their followers after observing the tagging policy and the tag (signal) realization. It is worth noting that the state variable remains hidden from the user throughout the game, as individuals lack the necessary resources to verify the authenticity of the content. In this context, the SNP aims to incentivize the agent's effort in supplying authentic content \emph{and} persuade the receiver to choose a desirable action. 

The action taken by the user results in a \emph{trend} (negative or positive about the post) in social media. To understand this notion, we consider a multi-type branching process (introduced later in Section \ref{sec:branching}). Denote by $N(t)$ the number of users who have just received the post with a negative comment at continuous time $t$ ($n$-type user). Similarly, $P(t)$ denotes the number of users who have received a positive comment ($p$-type user). After reading the received post, users forward it to some of their followers/friends with their own (either negative or positive) comments, producing ``offsprings'' (the new $n/p$-type users). The trend is measured through the proportion of negative comments over all the comments: $\eta(t)={N(t)}/(N(t)+P(t))$.  

In the persuasion literature \cite{kamenica11BP}, a key assumption is that the state distribution is revealed to the sender when deciding the tag realization under the designed policy. In the context of online misinformation circulation, this assumption is based on the premise that an SNP, as an institution, has the necessary resources to verify the authenticity of each post, as discussed in our prior work \cite{yang2023designing}. In this work, we address a more practical scenario by relaxing this assumption and considering that SNPs may have perceptions about the true state. These misperceptions could stem from the large volume of posts being made simultaneously while the SNP has limited capabilities and resources or from the error of misinformation detection. We introduce a mapping $d:\Omega\rightarrow \Delta(\Omega)$ that maps the actual state $\omega$ to a Borel probability measure $d(\cdot|\omega)$ (all sets in the models are endowed with Borel topology). From this measure, a new state $\omega'$ is sampled and becomes the SNP's misperception of the true state. In the misdetection scenario, $d(\omega'|\omega)$ gives the detection error rate of misclassifying $\omega$ as $\omega'$. For the rest of the paper, we refer to the SNP, the content provider, and the user as the sender, the agent, and the receiver, respectively. Of particular note is that the information structure regarding this misperception (who knows such a mapping) can lead to different treatments on equilibrium. Here we focus on the case where \textit{the sender, the agent, and the receiver are all aware of such misperception $d$}, which is motivated by the fact that SNPs may be increasingly required to disclose their misinformation detection and moderation due to transparency policies \cite{HR9126_118thCongress}.

To summarize the discussion above, the persuasion game is given by the tuple $\left\langle \Omega, \Sigma, \Lambda, \theta, d, \mathcal{A}, u_S, u_A, u_R \right\rangle$, where 
\begin{enumerate}[i)]
    \item $\Omega$ is the state space, and $\omega\in \Omega$ reflects how authentic the content of the post is;
    \item $\Sigma$ is the signal space of the sender, and $\sigma\in \Sigma$ denotes the tag associated with the post;
    \item $\Lambda$ is the action set of the agent, and each $\lambda\in \Lambda$ represents how much effort the agent exerts in producing trustworthy content;
    \item $\theta: \Lambda\rightarrow \Delta(\Omega)$ is the control function of the agent, whose effort $\lambda$ is turned into the state distribution $\theta(\cdot|\lambda)$ over the level of authenticity of the content $\Omega$;
    \item $d: \Omega\rightarrow\Delta(\Omega)$ is the sender's misperception, which maps the realized state $\omega$ to another state $\omega'$ following the distribution $\omega'\sim d(\cdot|\omega)$. This misperception is common knowledge.
    \item $\mathcal{A}$ is the action set of the receiver, which is a continuum $[0, 1]$, and $a \in \mathcal{A}$ denotes the probability of offering a positive comment;
    \item $\eta^*$ is the proportion of negative comment $\eta(t)$ as $t \rightarrow \infty$ obtained from the stabilized multi-type branching processes, which is related to the reputation of the agent and the impact of misinformation spreading; 
    \item $u_S: \Omega\times \mathcal{A}\rightarrow \mathbb{R}$, $u_A: \mathcal{A}\times \Lambda\rightarrow \mathbb{R}$, $u_R:\Omega\times \mathcal{A}\rightarrow\mathbb{R}$ are utility functions of the sender, the agent, and the receiver, respectively. The definitions of these utilities are as follows.
\end{enumerate}
A few remarks are in order. The state distribution $\theta(\cdot|\lambda)$ represents the misinformation circulation level, such as the percentage of fake news or misinformed posts on a social media platform \cite{konopliov2024fakenews}. The misperception $d(\cdot|\omega)$ reflects the false alarm rate, which stems from errors in human fact-checking \cite{andersen2020communicative} or AI-based classification systems \cite{islam2020deep,10631663}.

\textbf{The Receiver's Utility.}
To minimize the mismatch between the comment and the truth, the receiver's utility is $ u_R(\omega, a)=-(a-\omega)^2$. Suppose that the receiver believes that the state variable is subject to $\mu\in \Delta(\Omega)$, its best response under this belief is 
\begin{equation}
    a^*(\mu)=\argmax_{a\in [0,1]}\E_{\omega\sim \mu}[-(a-\omega)^2]=\E_{\mu}[\omega].
\label{eq:a-star}
\end{equation}  

\textbf{The Agent's Utility.}
The agent is concerned with the effort and its reputation measured through $\eta^*$ (the proportion of negative comments on its post). Denote by $c(\lambda)$ the cost induced by the effort $\lambda$; and by $r_A(a)=1-\eta^*(a)$ the agent's reputation when the receiver responds with $a$. Here, $\eta^*(a)$ is the proportion of negative feedback, and $1 - \eta^*(a)$ represents the proportion of positive comments that reflect the population level of ratings toward to produced content. In this case, the agent's utility is given by   
\begin{equation}
    u_A(a,\lambda)=r_A(a)-c(\lambda).
\label{eq:u_a}
\end{equation}

\textbf{The Sender's Utility.}
The sender’s goal is to mitigate
the influence of misinformation: 
the sender prefers more positive comments on authentic posts. Define
\begin{equation}
    u_S(\omega, a) = \omega(1 - \eta^*(a)),
\label{eq:u_si}
\end{equation} where $1 - \eta^*(a)$ represents the proportion of positive comments, and $\omega$ reflects the content's authenticity. This form implies that the sender benefits from a positive trend of authentic content, with the goal of reducing misinformation being implicit. 


The game unfolds in three stages. 1) In the first stage, the sender, aware of the misperception $d$, designs and commits to a signaling scheme $\pi: \Omega\rightarrow \Delta(\Sigma)$, specifying a condition distribution $\pi(\cdot|\omega')$ over the signal space. Note that both the misperception and the signaling scheme are known to the other two players.  2) Second, observing the signaling $\pi$, the agent chooses an private effort $\lambda$ to determine a favorable distribution over the state space $\theta(\cdot|\lambda)\in \Delta(\Omega)$. Note that the effort $\lambda$ is unobservable to both the sender and the receiver.  3) Finally, nature draws a state realization from $\theta(\cdot|\lambda)$, which is then distorted by $d(\cdot|\omega)$ and finally reveals distorted $\omega'$ to the sender. The sender then transmits a signal $\sigma$ (tag on the post) according to the commitment to the receiver, who, aware of both the signaling scheme and the misperception, chooses an action (determining how positively to comment on the post). A schematic illustration is provided in Fig. \ref{fig:per-mis}.

\subsection{Perfect Bayesian Equilibrium}\label{sec:PBE}
What distinguishes the introduced model from the classical Bayesian persuasion \cite{kamenica11BP} is that the receiver now does not explicitly acquire the prior distribution $\theta(\lambda)$, as $\lambda$ is unobservable. Hence, when the receiver acts, they must resort to a conjecture on the agent's action to update the posterior beliefs. This conjecture must be consistent with the agent's equilibrium choice, which naturally leads to the perfect Bayesian equilibrium (PBE) distinct from the subgame perfect equilibrium considered in the standard persuasion game \cite{tao23pot}.
In addition to the solution concept, another notable difference regards the priors. The prior, in this case, for three players is distorted by the sender's misperception $d$. 

We briefly state the PBE characterization, and details are presented in the ensuing subsection where a binary setting is considered. A PBE of the proposed persuasion game consists of a tagging policy $\pi$, the agent's effort $\lambda$, and a belief system $\{\mu_\sigma, \sigma\in \Sigma\}$\footnote{A belief system is a collection of posterior beliefs $\mu_\sigma$, and $\mu_\sigma$ denotes the belief when receiving signal $\sigma$.}, which satisfies the following properties:
\begin{enumerate}[i)]
    \item given a signaling $\pi$ (sender) and a belief system $\{\mu_\sigma, \sigma\in \Sigma\}$ (receiver), the agent's effort $\lambda$ maximizes their expected utility, i.e.,
    \begin{align}
        &\lambda =\argmax \sum_{\omega}\theta(w|\lambda)\sum_{\sigma, \omega^\prime}d(\omega^\prime|\omega)\pi(\sigma|\omega^\prime)u_A(\mu_\sigma, \lambda)\label{eq:lambda-argmax},\\
        & u_A(\mu_\sigma, \lambda)= u_A(a^*(\mu_\sigma), \lambda),\nonumber
    \end{align}
    \item the receiver's belief is consistent with the agent's effort $\lambda$ and the signaling $\pi$, i.e.,
    \begin{align}
        &\mu_\sigma=\frac{d^\tp \pi(\sigma|\cdot)\odot   \theta(\cdot|\lambda)}{\langle d^\tp\pi(\sigma|\cdot) ,  \theta(\cdot|\lambda) \rangle}, \label{eq:consistency_d}\\
        & \pi(\sigma|\cdot)= [\pi(\sigma|\omega_1),\ldots, \pi(\sigma|\omega_N)]\in \mathbb{R}^{|\Omega|},\\
        & \theta(\cdot|\lambda)= [\theta(\omega_1|\lambda), \ldots, \theta(\omega_N|\lambda)]\in \mathbb{R}^{|\Omega|},
    \end{align}
    where $\odot$ denotes the pointwise product and the distorted prior  by the sender’s misperception
    can be understood as a matrix publication: $d\in \R^{|\Omega|\times |\Omega|}$ with $d_{ij}=d(\omega_i|\omega_j)$; $\theta_j=\theta(\omega_j|\lambda)\in \R^{|\Omega|}$,
    \item the signaling maximizes the sender's expected utility, i.e.,
    \begin{align}
    \label{eq:sender-max}
        \pi\in\argmax \sum_{\omega}\theta(\omega|\lambda)\sum_{\sigma, \omega'}d(\omega'|\omega)\pi(\sigma|\omega')u_S(a^*(\mu_\sigma), \omega ).
    \end{align}
\end{enumerate}

\subsection{Binary-State Case}
We use a binary case study for simplicity, where the state space consists of two elements $\Omega=\{0,1\}$ with $0$ indicating the content contains misinformation while $1$ represents the content is authentic. Hence, the signal space is also assumed to be binary: $\Sigma=\{0,1\}$, where $0$ and $1$ denote the ``fake'' and ``real'' tags, respectively. Since the state space is binary, the corresponding prior distribution of the authenticity of the content lives in the simplex spanned by $\theta_0=[1,0]$ and $\theta_1=[0,1]$. Therefore, we assume that the effort $\lambda$ spent by the agent is a scalar from $[0,1]$, and the resulting prior distribution is the convex combination of $\theta_0$ and $\theta_1$: $\theta(\lambda)=(1-\lambda)\theta_0+\lambda \theta_1$. In this binary setup, the misperception $d$ is given by a 2-by-2 stochastic matrix: 
\begin{align}
    d=\begin{bmatrix}
    1-\varepsilon_0 & \varepsilon_1\\
    \varepsilon_0 & 1-\varepsilon_1
\end{bmatrix},
\end{align}
where $\varepsilon_0$ and $\varepsilon_1$ can be interpreted as the false alarm rates under $\omega=0$ and $\omega=1$, respectively.
\begin{assumption}
    For the false alarm rates $\varepsilon_0,\varepsilon_1 \in \mathbb{R}_{\ge 0}$, we assume $\varepsilon_0 + \varepsilon_1 < 1$.
\label{assump:eps}
\end{assumption}

As the state space is finite, the players' strategies are finite-dimensional vectors, and hence, we can ``vectorize'' our analysis so that convex analysis tools can be utilized. Let $v_A(\mu)=r_A(a^*(\mu))$ denote the agent's payoff under the receiver's belief $\mu$, and $\Bar{v}_A^d(\omega|\pi):=\sum_{\sigma}\sum_{\omega'}d(\omega'|\omega)\pi(\sigma|\omega')v_A(\mu_\sigma)$ denote the agent's expected payoff conditional on the generated state $\omega$ under the signaling $\pi$ considering the misperception $d$. Then, let $\Vec{v}_A^d(\pi)$ be the corresponding vector: $\Vec{v}_A^d(\pi)=[\Bar{v}_A^d(0|\pi),\Bar{v}_A^d(1|\pi) ]$. Similarly, we have the following notations for the sender. Given the receiver's belief $\mu$, the sender's expected payoff is denoted by $v_S(\mu):=\E_{\omega\sim \mu}[u_S(a^*(\mu),\omega)]$. Let $\Bar{v}^d_S(\omega|\pi)=\sum_{\sigma}\sum_{\omega'} d(\omega'|\omega)\pi(\sigma|\omega')v_S(\mu_\sigma)$ and $\Vec{v}_S(\pi)^d:=[\Bar{v}^d_S(0|\pi), \Bar{v}^d_S(1|\pi)]$.

Additionally, we impose the following customary assumption \cite{kamenica11BP, boleslavsky18moral-hazard} on the cost of effort to ensure that the agent's equilibrium problem is well-behaved. This assumption maintains generally in our analysis, with the numerical study specifying the cost as $k\lambda^2$, where $k \in \mathbb{R}_{\ge0}$ is a parameter.
\begin{assumption}
    For the agent's utility given by (\ref{eq:u_a}), we assume that $r_A(\cdot)$ is non-negative and bounded, and $c(\cdot)\in C^2$ is strictly increasing and convex. In addition, $c(0)=\nabla c(0)=0$, and $\nabla c(1)>1$.
\label{assump:cost}
\end{assumption}

To characterize the PBE in the proposed model, we need the backward induction, i.e., first analyzing the optimality actions of the receiver, then the agent, and finally the sender. To begin with, the receiver's best response (comment) under the belief $\mu$ is given by \eqref{eq:a-star}. The best-response $a^*(\mu_\sigma)$ then affects the spread of misinformation in social media through branching processes presented in Section \ref{sec:branching}.

\section{Content Spreading Through Multi-type Branching Process}
This section treats the spread of misinformation through branching processes. Specifically, we focus on the evolution of the trend $\eta(t)$, the proportion of negative comments, as the receiver forwards the post to others. One key finding is that the evolutionary dynamics of $\eta(t)$ under the branching process stabilizes in the limit, and the receiver's belief completely determines the stationary point $\eta^*$. 

\subsection{Multi-type Branching Processes}\label{sec:branching}
Suppose that the number of the receiver's friend $M$ is independent and identically distributed with expectation $\mathbb{E}[M]=m_{M}$ and is finite. The receiver shares the post with $Bin(M, q)$ friends, where $q \in [0, 1]$ represents the impact or attractiveness of the post (assumed to be constant). Hence, the number of ``offspring'' (friends receiving the sharing) of the receiver, denoted by $\xi$, is subject to a binomial distribution: $\xi \sim Bin(M, q)$ with $\mathbb{E}[\xi]=m_{M}\cdot q:= m$.

Let $t$ denote the continuous time, and let $t_i$ represent the time at which the $i$-th user ``wakes up'', meaning that this individual becomes active on an SNP and is ready to share the post. Denote by $N_i = N(t_i^+)$, $P_i = P(t_i^+)$, and $Z_i = N_i + P_i$, where $t_i^+$ represents the right-hand limit of $t_i$. This enables our analysis of the branching process at transition times (i.e., when a user wakes up) by discretizing the continuous $N(t)$ and $P(t)$ into their corresponding discrete counterparts, $N_i$ and $P_i$, thereby forming $Z_i$. Moreover, let $\xi_i \overset{i.i.d.}{\sim} \text{Bin}(M, q)$. Then, if the $n$-type receiver (who receives negative comments) wakes up at $t_{i+1}$, then 
\begin{equation}
  \begin{aligned}
    N_{i+1} &= N_{i}-1 + \textbf{1}_{n} \xi_i,\\
    P_{i+1} & = P_{i} + \textbf{1}_{p} \xi_i,
\end{aligned}  \label{eq:x-wake}
\end{equation}
and if the $p$-type receiver wakes up, 
\begin{equation}
    \begin{aligned}
    N_{i+1} &= N_{i} + \textbf{1}_{n} \xi_i,\\
    P_{i+1} &= P_{i} -1 + \textbf{1}_{p} \xi_i.
\end{aligned}\label{eq:y-wake}
\end{equation}
 where the indicator function $\textbf{1}_{n}$ means that the receiver makes a negative comment while $\textbf{1}_{p}$ indicates the opposite (the positive comment). The total population is updated by $Z_{i+1} = Z_{i} -1 + \xi_i$.

The probability of a receiver who receives the post with a negative comment also commenting negatively can be captured by a negative-to-negative factor $\alpha_{nn}(\sigma)$, which depends on the tag $\sigma$. Similarly, the positive-to-negative factor $\alpha_{pn}(\sigma)$ represents the probability of a receiver leaving a negative comment after receiving and viewing the post with a positive comment. As the receiver's comment only depends on the belief $\mu_\sigma$ [see the best response in \eqref{eq:a-star}], $\alpha_{nn}(\sigma)=\alpha_{pn}(\sigma)=1-a^*(\mu_\sigma)=1-\E_{\mu_\sigma}[\omega]$. That is, a higher $E_{\mu_\sigma}[w]$ indicates greater confidence from the receiver regarding the authenticity of the post's content, making them less likely to leave a negative comment.

\subsection{Stochastic Approximation Analysis}

To analyze the limit trend of the process, we apply stochastic approximation \cite{SA_for_BP} and consider the continuous-time dynamics of the multi-type branching process. Since there are only two types in the branching process, it suffices to consider the dynamics of the total population and that of the $n$-type. Toward this end, let $\Bar{Z}_i = \frac{Z_i}{i}$, $\Bar{N}_i = \frac{N_i}{i}$, and $\gamma_i = \frac{1}{i+1}$, and then we aggregate the branching equations in \eqref{eq:x-wake} and \eqref{eq:y-wake}, leading to the following:  
\begin{equation}
    \begin{aligned}
    \Bar{Z}_{i+1} = \Bar{Z}_i &+ \gamma_{i}\big(\xi_i -1-\Bar{Z}_i\big) \textbf{1}_{\{\Bar{Z}_i > 0\}}, \\
    \Bar{N}_{i+1} = \Bar{N}_i &+ \gamma_{i}\big[\textbf{1}_{\{n-wakes\}}\big(\textbf{1}_{n} \xi_i -1\big)\\
    &+ \textbf{1}_{\{p-wakes\}}\textbf{1}_{n} \xi_i - \Bar{N}_i \big] \textbf{1}_{\{\Bar{Z}_i > 0\}},
\end{aligned}\label{eq:discrete-ode}
\end{equation}
where $\mathbb{E}[\textbf{1}_{\{n-wakes\}}]=\frac{\Bar{N_i}}{\Bar{Z_i}}$, $\mathbb{E}[\textbf{1}_{\{p-wakes\}}]=1-\frac{\Bar{N_i}}{\Bar{Z_i}}$ indicate the probabilities of a receiver of $n$-type and $p$-type wakes up. Let $\Bar{N}_0 = N_0$, $\Bar{Z}_0 = N_0+P_0$ be the initial conditions. As the discrete-time trajectory of \eqref{eq:discrete-ode} is an asymptotic pseudo-trajectory of the continuous-time system in \eqref{eq:cont-ode} \cite{SA_for_BP}, the two systems share the same limiting behavior. Hence, we arrive at \Cref{prop:sa}.
\begin{equation}
    \begin{aligned}
    \Dot{z} &= h^z(z, n) = (m-1-z)\textbf{1}_{\{z > 0\}}, \\
    \Dot{n} &= h^n(z, n) = \big[\eta\big(\alpha_{nn}(\sigma) \cdot m-1\big)\\ &+ (1-\eta)\alpha_{pn}(\sigma) \cdot m - n \big]\textbf{1}_{\{z > 0\}}, \eta = \frac{n}{z}
\end{aligned}\label{eq:cont-ode}
\end{equation}
\begin{proposition}
    Consider $\mathbb{E}[M^2] < \infty$ in the multi-type branching process, the $\{\Bar{Z}_{i}\}, \{\Bar{N}_{i}\}$ sequences converge to $\Bar{Z}^*, \Bar{N}^*$ almost surely, where $\Bar{Z}^* = m-1$ and $\Bar{N}^*=\eta^*(\sigma) \Bar{Z}^*$ with $\eta^*(\sigma) = \frac{\alpha_{pn}(\sigma)}{1-\alpha_{nn}(\sigma)+\alpha_{pn}(\sigma)}$ are solutions to \eqref{eq:cont-ode}.
\label{prop:sa}
\end{proposition}

The proof for the above proposition follows \cite{BPBP_tagging}. Note that $\eta^*(\sigma)$ and $\eta^*(a)$ can be used interchangeably because the receiver decides an action $a$ based on the posterior belief $\mu_\sigma$ with respect to the tag $\sigma$. Since the receiver's comment only depends on the belief, we can characterize the limiting trend under tag $\sigma$ by the following statement.
\begin{corollary}
   As $\alpha_{pn}(\sigma) = \alpha_{nn}(\sigma) = 1-\E_{\mu_\sigma}[\omega]$, then the proportion of negative comments $\eta^*(\sigma)=\eta^*(a(\mu_\sigma))=\alpha_{pn}(\sigma)=1-\E_{\mu_\sigma}[\omega]$.
\label{remark}
\end{corollary}

\subsection{Optimality Conditions under Stable Branching}
Given the receiver's best response $a^*(\mu_\sigma)$ and the stabilized branching process result, we can now simplify the agent's problem, as the trend $\eta^*(\sigma)$ admits a simple formula. Since $\eta^*(a)=1-\E_{\mu}[\omega]$ from \Cref{remark}, we notice that $v_A(\mu)=r_A(a^*(\mu))=1-\eta^*(a)=\E_{\mu}[\omega]=\mu(1)$, which is linear in $\mu(1)$. In the binary-state case, the belief $\mu_\sigma$ is uniquely determined by its second entry $\mu(1)$. Hence, the following discussion will treat $\mu_\sigma$ as a scalar. The same treatment also applies to the prior $\theta$. The agent's optimality conditions under the signaling in \eqref{eq:lambda-argmax} can be rewritten as 
\begin{align*}
    \max_{\lambda\in [0,1]}\langle \theta(\lambda), \Vec{v}_A^d(\pi)  \rangle -c(\lambda).
\end{align*}
Given the linearity of the first term and the convexity of the second term, the problem is an unconstrained convex optimization problem. Therefore, taking the first-order derivative of the objective function leads to the following first-order condition for optimality \cite{boyd2004convex}: 
\begin{align}
    \langle \theta_1-\theta_0, \Vec{v}_A^d(\pi) \rangle =\nabla c(\lambda), \label{eq:agent-opt}
\end{align} As later shown in the ensuing section, the agent's marginal cost $\nabla c$ plays a significant part in the feasibility of the sender's information structures. 

Since $\eta^*(a)=1-\E_{\mu}[\omega]$ and then $1-\eta^*(a)=\E_{\mu}[\omega]$, the sender's expected utility under the belief $\mu$ is $v_{S}(\mu)=\E^2_{\mu}[\omega]$, which is convex in $\mu$ and non-negative. In the binary-state case, $v_S(\mu)=\mu^2$. 
Hence, the sender's problem is given by 
\begin{equation}
    \begin{aligned}
    \max_{\pi, \lambda} & \langle \theta(\lambda), \Vec{v}_S^d(\pi) \rangle \\
    \text{s.t. } & \langle \theta_1-\theta_0, \Vec{v}_A^d(\pi) \rangle =\nabla c(\lambda), \\
    &\mu_\sigma=\frac{d^\tp \pi(\sigma|\cdot)\odot   \theta(\cdot|\lambda)}{\langle d^\tp\pi(\sigma|\cdot) ,  \theta(\cdot|\lambda) \rangle}.
\end{aligned}\label{eq:sender-old}
\end{equation}
Note that the agent's decision variable $\lambda$ also appears in the maximization, as we assume that the tie breaks in favor of the sender should there exist multiple effort level $\lambda$ satisfying the first constraint in \eqref{eq:sender-old}. It should be noted that both the objective and the first constraint are linear in $\pi$ and admit a linear programming formulation \cite{tao22bp}. However, the challenge lies in the second constraint, which is the consistency requirement in \eqref{eq:consistency_d} and involves division operation, leading to a highly nonlinear programming problem. To simplify our analysis, the proposition in the following section \ref{sec:plausibility} uses Bayesian Plausibility to transform the sender's problem into the posterior belief space.

\subsection{Finite-State Persuasion Game}
Before concluding this section, we briefly touch upon the generic persuasion model with finite state, signal, and action space. The assumption of finite discrete spaces is made for the purpose of demonstrating the complexity in computing the perfect Bayesian equilibrium. In contrast, the binary case admits an elegant Lagrangian approach to characterize the optimal solution without solving the optimization problem as presented in the ensuing section. The developed Lagrangian approach also lends itself to the generic convex utility function (\Cref{prop:optimal-is}), and the binary case considered in this work provides a simple and illustrative example. 

To facilitate the discussion, we first ``vectorize'' the key components in the persuasion game model as in the binary case. Let the state, signal, and receiver action space be $\Omega=\{\omega_i\}_{i\in [P]}$, $\Sigma=\{\sigma_i\}_{i\in [Q]}$, and $\mathcal{A}=\{a_i\}_{i\in [K]}$, respectively, where $[P]\triangleq\{1,2,\ldots, P\}$. To further simplify the exposition, we fix the agent's action $\lambda$ and represent the state distribution as a diagonal matrix $\Theta\triangleq \operatorname{diag}\{\theta_1, \theta_2, \ldots, \theta_P\}$, where $\theta_i\triangleq \theta(\omega_i|\lambda)$. The misdetection can also be expressed as a stochastic matrix: $D_{ij}\triangleq d(w'_j|w_i)$, and $D\mathds{1}=\mathds{1}$. Similarly, the sender's signaling $\pi$ takes the following stochastic matrix form: $\Pi_{mn}\triangleq \pi(\sigma_n|w_m)$. 

Upon receiving the signal $\sigma_n\in \Sigma$, the receiver derives the Bayesian posterior belief following the consistency in \eqref{eq:consistency_d}. Let $\mu_{mn}$ be the receiver's belief of state $\omega_m$ after observing $\sigma_n$ (i.e., the $m$-th entry of $\mu_n$ in \eqref{eq:consistency_d}). Then, we arrive at
\begin{equation}
\label{eq:bayes-post}
    \mu_{mn}=\frac{\theta_m \sum_{m'\in [P] }D_{mm'}\Pi_{m'n}}{\sum_{m\in [P]}\theta_m\sum_{m'\in [P]} D_{mm'}\Pi_{m'n}}.
\end{equation}
Define the belief system $\{\mu_n\}_{n\in [Q]}$ as $U\triangleq[\mu_{mn}]\in \R^{P\times Q}$, and translating \eqref{eq:bayes-post} into matrix presentation, one obtains
\begin{equation}
\label{eq:consistency-matrix}
    U=\Theta D\Pi \oslash (\mathds{1}\mathds{1}^\tp \Theta D\Pi),
\end{equation}
where $\oslash$ denotes the Hadamard (entry-wise) division. Based on the posterior belief $\mu_n$, the receiver decides an action, which we model as a non-negativbe stochastic matrix $A=[A_{nk}]\in \R_{\geq 0}^{Q\times K}$, where $A_{nk}$ denotes the probability of choosing $a_k$ upon receiving $\sigma_n$ (inducing belief $\mu_n$). 

We now utilize the matrix inequality to characterize the receiver's best response under the induced belief. Let $S=[S_{km}]\in \R^{K\times P}$ and $R=[R_{km}]\in\R^{K\times P}$ be the matrix representations of the sender's and receiver's utility function, respectively, where their $(k,m)$-entry denotes the utilities under state $\omega_m$ and action $a_k$. Suppose the receiver's response policy $A$ is the best response, then for any belief $\mu_n$, $n\in [Q]$, we have the following inequality hold for any other stochastic matrix $A'$:
\begin{equation*}
    \sum_{m\in [P]} \mu_{mn} \sum_{k\in [K]} A_{nk}R_{km}\geq \sum_{m\in [P]} \mu_{mn} \sum_{k\in [K]} A'_{nk}R_{km},
\end{equation*}
which suggests that the policy $A$ brings up higher expected utility under any belief. When translated into compact matrix representations, the above inequalities (one for each $n\in [Q]$)  lead to the matrix inequality in \eqref{eq:br-matrix}. With a light abuse of notation, we denote by $\operatorname{diag}(W)$ the vector composed of diagonal entries of matrix $W$ and by $\succeq$ the entry-wise $\geq$ relation between two vectors.
\begin{equation}
    \label{eq:br-matrix}
    \operatorname{diag}(ARU) \succeq \operatorname{diag}(A'RU), \forall A'\in \R_{\geq 0}^{Q\times K}, A'\mathds{1}=\mathds{1}.
\end{equation}

Employing the same argument, we derive the sender's optimal signaling. Suppose the optimal solution to \eqref{eq:sender-max} $\Pi$ satisfies the following inequality
\begin{align*}
    &\sum_{m\in [P]}\theta_m \sum_{n\in [Q], m'\in [P]}D_{mm'}\Pi_{m'n}A_{nk}S_{km}\\
    &\geq \sum_{m\in [P]}\theta_m \sum_{n\in [Q], m'\in [P]}D_{mm'}\Pi'_{m'n}A_{nk}S_{km},
\end{align*}
for any other stochastic matrix $\Pi'$. Similarly, the above inequalities admit a compact matrix representation. Denote by $\operatorname{Tr}(\cdot)$ the trace operator, we arrive at
\begin{equation}
    \label{eq:sender-max-matrix}
    \operatorname{Tr}(\Theta D \Pi A S)\geq \operatorname{Tr}(\Theta D \Pi' A S), \forall \Pi'\in \R_{\geq 0}^{P\times Q}, \Pi'\mathds{1}=\mathds{1}.
\end{equation}
Finally, summarizing \eqref{eq:consistency-matrix}, \eqref{eq:br-matrix}, and \eqref{eq:sender-max-matrix}, we can define the perfect Bayesian equilibrium in matrix form as in \Cref{def:pbe-matrix}.
\begin{definition}[Perfect Bayesian Equilibrium in Matrix]
\label{def:pbe-matrix}
    For a finite persuasion game, a triple of matrices $(\Pi, A, U)$ is a perfect Bayesian equilibrium if it satisfies 
    \begin{equation}
    \label{eq:pbe-matrix}
        \begin{aligned}
            & \operatorname{Tr}(\Theta D \Pi A S)\geq \operatorname{Tr}(\Theta D \Pi' A S),\\
            & \forall \Pi'\in \R_{\geq 0}^{P\times Q}, \Pi'\mathds{1}=\mathds{1}, \forall \Pi'\in \R_{\geq 0}^{P\times Q}, \Pi'\mathds{1}=\mathds{1},\\
            & \operatorname{diag}(ARU) \succeq \operatorname{diag}(A'RU), \forall A'\in \R_{\geq 0}^{Q\times K}, A'\mathds{1}=\mathds{1},\\
            & U=\Theta D\Pi \oslash (\mathds{1}\mathds{1}^\tp \Theta D\Pi).
        \end{aligned}
    \end{equation}
\end{definition}

We now comment on the computation complexity of solving the matrix inequality in \eqref{eq:pbe-matrix}. Prior works established that solving for the equilibrium signaling $\Pi$ is NP-hard \cite{rubinstein15hardness,bhaskar16hardness,dughmi19hardness,tao22bp}. Furthermore, \cite{tao23pot} developed a two-stage bilinear programming method for equilibrium computation. However, the mathematical programming method can only handle a subset of equilibrium: non-degenerate belief-dominant perfect Bayesian equilibrium. The key message of our work is that the Lagrangian conveys sufficient information to determine the equilibrium solution without exact computation if the utility function is convex with respect to the belief, as established in \Cref{prop:optimal-is}.
\section{Perfect Bayesian Equilibrium Characterization: A Lagrangian Approach}

\subsection{Bayesian Plausibility} \label{sec:plausibility}

Bayesian plausibility \cite{kamenica11BP} serves as a crucial sanity check for any information structure: all the posterior beliefs generated by the observed signals must align with the prior distribution within that structure. The following proposition reformulates the sender's problem by shifting the focus from a tagging policy $\pi$ to a distribution over posteriors $\tau^d\in \Delta(\Delta(\Omega))$ as the decision variable.


\begin{proposition}[Bayesian Plausibility]
    Given an effort $\lambda$, there exists a signaling  $\pi$ satisfying the conditions in problem \eqref{eq:sender-old} if and only if there exists a distribution over posteriors $\tau^d \in \Delta(\Delta(\Omega))$ such that 
    \begin{align*}
    &\E_{\tau^d}[\mu] = \theta(\lambda),\\
    &\E_{\tau^d} \left[\E_{\mu}[\nabla \log \theta(\lambda) ]v_A(\mu)\right]=\nabla c(\lambda).
    \end{align*}
\end{proposition}

\begin{proof}
    We first need to prove the equivalence between the signaling mechanism $\pi$ and the distribution $\tau^d$. Without loss of generality, assume that for each signal $\sigma\in\Sigma$, the receiver has a distinct posterior belief $\mu_\sigma$. Starting from $\pi$, and fixing $\lambda$, the probability of generating $\mu_\sigma$ is 
    \begin{equation*}
        \tau^d(\mu_\sigma)= \sum_{\omega, \omega'}\pi(\sigma|\omega')d(\omega'|\omega)\theta(\omega|\lambda)=\langle d^\tp\pi(\sigma|\cdot), \theta(\lambda) \rangle.
    \end{equation*} From the posterior belief in \eqref{eq:consistency_d} and the definition of $\tau^d$, we have
    \begin{equation*}
        \pi(\sigma|\cdot) = \tau^d(\mu_\sigma)(d^\tp)^{-1}(\mu_\sigma\oslash \theta(\cdot|\lambda)),
    \end{equation*}
    where we assume that $d$ is nonsingular. The nonsingularity is easy to satisfy, as the determinant $\operatorname{det}(d)=1-\varepsilon_0-\varepsilon_1$ is nonzero according to Assumption \ref{assump:eps}. 
    Then, we have  
    \begin{align*}
       &  d^\tp \pi(\sigma|\cdot) = \tau^d(\mu_\sigma)(\mu_\sigma\oslash \theta(\cdot|\lambda))\\
       \Leftrightarrow\quad  & \sum_{\sigma} d^\tp \pi(\sigma|\cdot) \odot \theta(\cdot|\lambda) =\sum_\sigma \tau^d(\mu_\sigma) \mu_\sigma
    \end{align*}
    Using the distributivity of matrix multiplication, 
    the left-hand side is indeed 
    $$d^\tp(\sum_\sigma \pi(\sigma|\cdot))\odot \theta(\cdot|\lambda)=d^\tp {1}\odot \theta(\cdot|\lambda)= \theta(\cdot|\lambda),$$
    where the last equality follows the left stochasticity of $d$. Therefore,  $\E_{\tau^d}[\mu]=\theta(\lambda)$, which proves the first equality in the proposition. Note that the posterior distribution $\tau^d$ associated with $\pi$ is called the Bayesian-plausible distribution in the literature \cite{kamenica11BP}, and that the first equality shows Bayesian plausibility holds with respect to the original prior $\theta(\lambda)$ instead of the distorted one. 

    To recover the agent's optimality condition (also called incentive-compatibility constraint), consider the constraint:
    \begin{align*}
        & \langle \theta_1- \theta_0,   \Vec{v}^d_A(\pi) \rangle \\
        &= \sum_{\omega}\left(\sum_{\sigma}\sum_{\omega'}d(\omega'|\omega)\pi(\sigma|\omega')v_A(\mu_\sigma)\right)(\theta_1(\omega)-\theta_0(\omega))\\
        & =\sum_{\omega}\left( \frac{\sum_{\sigma}\tau^d(\mu_\sigma)\mu_\sigma(\omega)}{\theta(\omega|\lambda)}v_A(\mu_\sigma)\right)(\theta_1(\omega)-\theta_0(\omega))\\
        & = \E_{\tau^d}[\E_{\mu}[\nabla_\lambda \log \theta(\omega|\lambda)]v_A(\mu)] = \nabla c(\lambda),
    \end{align*} which proves the second equality in the proposition.
\end{proof}

Hence, by letting $f(\mu)=\E_{\mu}[\nabla_\lambda \log \theta(\omega|\lambda)] v_A(\mu)-\nabla c(\lambda)$, the sender's problem can be rewritten as 
\begin{align}
\max_{\tau^d\in\Delta(\Delta(\Omega)), \lambda} & \E_{\tau^d}[v_S(\mu)],\label{eq;sender-max}\\
 \text{s.t. } & \E_{\tau^d}[\mu]=\theta(\lambda),\label{eq:bp}\\
& \E_{\tau^d} [f(\mu)]=0,\label{eq: ic}
\end{align}
where \eqref{eq:bp}, referred to as the Bayesian plausibility constraint (BP), corresponds to the consistency in \eqref{eq:consistency_d}; \eqref{eq: ic}, referred to as the incentive-compatibility constraint (IC), rephrases the agent's optimality condition in \eqref{eq:agent-opt}.

\subsection{Feasible Posterior Beliefs}\label{sec:feasible_mu}
It is worth noting that due to the sender's misperception $d$ with false alarms $\varepsilon_0$ and $\varepsilon_1$, the posterior beliefs $\mu$  can not span the entire $[0, 1]$. To see this,
consider the binary-state case,
\begin{align*}
    d^\tp \pi(\sigma|\cdot)&=\begin{bmatrix}
    1-\varepsilon_0 & \varepsilon_0\\
    \varepsilon_1 & 1-\varepsilon_1
\end{bmatrix}\begin{bmatrix}
    \pi(\sigma|0) \\
    \pi(\sigma|1)
\end{bmatrix}\\
&=\begin{bmatrix}
    (1-\varepsilon_0)\pi(\sigma|0)+ \varepsilon_0 \pi(\sigma|1)\\
    \varepsilon_1 \pi(\sigma|0)+(1-\varepsilon_1)\pi(\sigma|1)
\end{bmatrix},
\end{align*}
\begin{align*}
    \mu_\sigma &= \frac{d^\tp \pi(\sigma|\cdot) \odot \theta(\cdot|\lambda)}{\tau^d(\mu_\sigma)}\\
    &=\begin{bmatrix}
    \frac{(1-\lambda)[(1-\varepsilon_0)\pi(\sigma|0)+ \varepsilon_0 \pi(\sigma|1)]}{(1-\lambda)[(1-\varepsilon_0)\pi(\sigma|0)+ \varepsilon_0 \pi(\sigma|1)]+\lambda[\varepsilon_1 \pi(\sigma|0)+(1-\varepsilon_1)\pi(\sigma|1)]}\\
    \frac{\lambda[\varepsilon_1 \pi(\sigma|0)+(1-\varepsilon_1)\pi(\sigma|1)]}{(1-\lambda)[(1-\varepsilon_0)\pi(\sigma|0)+ \varepsilon_0 \pi(\sigma|1)]+\lambda[\varepsilon_1 \pi(\sigma|0)+(1-\varepsilon_1)\pi(\sigma|1)]}
\end{bmatrix}.\label{eq:mu_s_d}
\end{align*} In this case, $\mu=1$ only when $\lambda=1$. Hence, for given values of $\lambda, \varepsilon_0$, and $ \varepsilon_1$, $\mu$ can not span the entire range of $[0, 1]$.

As proved by \cite{WU2023105763}, more information signaling leads to more dispersed beliefs. We identify the feasible space for posterior beliefs through ``fully informative'' signaling, which truthfully and deterministically reveals the content authenticity. Let $\overline{\pi}$ represent the fully informative tagging, where $\overline{\pi}(0|0)=\overline{\pi}(1|1)=1, \overline{\pi}(1|0)=\overline{\pi}(0|1)=0$. In this scenario, the receiver, upon receiving the tag $\sigma$, is certain about the sender's perceived authenticity: the post is either fake $0$ or authentic $1$. 
When the received tag $\sigma=0$, denote $\mu_{\sigma=0} = [1-\underline{\mu}, \ \underline{\mu}]^\tp$,
\begin{align}
    \mu_{\sigma=0}=\begin{bmatrix}
    1-\underline{\mu}\\
    \underline{\mu}
\end{bmatrix}=\begin{bmatrix}
    \frac{(1-\lambda)(1-\varepsilon_0)}{(1-\lambda)(1-\varepsilon_0)+\lambda\varepsilon_1} \\
    \frac{\lambda \varepsilon_1}{(1-\lambda)(1-\varepsilon_0)+\lambda\varepsilon_1} 
\end{bmatrix}, 
\end{align}with $\overline{\tau}^d(\mu_{\sigma=0})=\overline{\tau}^d(\underline{\mu})=(1-\lambda)(1-\varepsilon_0)+\lambda\varepsilon_1$ under fully-informative tagging policy $\overline{\pi}$. While the received tag $\sigma=1$, denote $\mu_{\sigma=1} = [1-\overline{\mu} \ \overline{\mu}]^\tp$, then we have 
\begin{align}
    \mu_{\sigma=1}=\begin{bmatrix}
    1-\overline{\mu} \\
    \overline{\mu} 
\end{bmatrix}=\begin{bmatrix}
    \frac{(1-\lambda)\varepsilon_0}{(1-\lambda)\varepsilon_0+\lambda(1-\varepsilon_1)} \\
    \frac{\lambda(1-\varepsilon_1)}{(1-\lambda)\varepsilon_0+\lambda(1-\varepsilon_1)} 
\end{bmatrix}. 
\end{align} with $\overline{\tau}^d(\mu_{\sigma=1})=\overline{\tau}^d(\overline{\mu})= (1-\lambda)\varepsilon_0+\lambda(1-\varepsilon_1)$. Then, considering the receiver's belief $\mu$ resulting from an arbitrary tagging policy, we can observe that $$0 \leq \underline{\mu} \leq \mu \leq \overline{\mu} \leq 1.$$
By noticing this, we denote $\mu \in [\underline{\mu}, \overline{\mu}]$ to represent the feasible posterior belief spaces.

\begin{proposition}
    Under fully informative tagging in the binary-state case, $\underline{\mu}$ is convex and increasing in $\lambda$, while $\overline{\mu}$ is concave and increasing in $\lambda$.
\label{prop:increasing_mu}
\end{proposition}
\begin{proof}
    The proof is provided in Appendix \ref{app:increasing_mu}.
\end{proof}

\subsection{The Lagrangian Characterization}

With Bayesian plausibility, the sender's problem becomes equality-constrained nonlinear programming, which naturally prompts one to consider the Lagrange multiplier method. In what follows, we present a PBE characterization through the lens of Lagrangian. The discussion begins with the feasible domain of the maximization in \eqref{eq;sender-max}.  
\begin{proposition}[Implementable Effort, Feasible Condition]
In the binary-state model, let $\bar{\lambda}$ be the value such that $\nabla c(\lambda)=(\theta_1-\theta_0)D(\overline{\mu}-\underline{\mu})$, where $D=\operatorname{det}(d)=1-\varepsilon_0-\varepsilon_1$. Then, $\lambda$ is feasible if and only if $\lambda\leq\bar{\lambda}$.
\label{prop:feasible}
\end{proposition}
\begin{proof}
We begin with the necessity. In the binary-state case, the IC constraint reduces to 
    \begin{align*}
        (\theta_1-\theta_0)(\bar{v}_A^d(1)-\bar{v}_A^d(0))=\nabla c(\lambda),
    \end{align*}
    where $\Bar{v}_A^d(\omega|\pi):=\sum_{\sigma}\sum_{\omega'}d(\omega'|\omega)\pi(\sigma|\omega')v_A(\mu_\sigma)$. Then, let $D=\operatorname{det}(d)=1-\varepsilon_0-\varepsilon_1$ and note that $v_A(\mu)=\mu \in [\underline{\mu},\overline{\mu}]$, 
    \begin{align*}
        & \Bar{v}_A^d(1) - \Bar{v}_A^d(0)\\
        &= (\varepsilon_1-(1-\varepsilon_0))\pi(0|0)\mu_0 + (1-\varepsilon_1-\varepsilon_0) \pi(0|1)\mu_0\\
        & \quad + (\varepsilon_1-(1-\varepsilon_0))\pi(1|0)\mu_1 + (1-\varepsilon_1-\varepsilon_0) \pi(1|1)\mu_1\\
        &= D(\pi(1|1)+\pi(0|0)-1)(\mu_1-\mu_0). 
    \end{align*} Since $(\pi(1|1)+\pi(0|0)-1) \leq 1$, $\Bar{v}_A^d(1) - \Bar{v}_A^d(0)$ never exceeds $D(\overline{\mu}-\underline{\mu})<1$.
    Hence, $(\theta_1-\theta_0)D(\overline{\mu}-\underline{\mu}) \geq \nabla c(\lambda)$. As $c(\cdot)$ is strictly increasing, $\nabla c(\lambda)> \nabla c(\bar{\lambda})=(\theta_1-\theta_0)D(\overline{\mu}-\underline{\mu})$, for $\lambda>\bar{\lambda}$, which means $\lambda$ is not IC.

    For sufficiency, consider $\lambda\in (0, \bar{\lambda}]$, and $\theta(\lambda)=(1-\lambda, \lambda)$. We construct a Bayesian-plausible hybrid $\tau^h$ as follows. $\operatorname{supp}(\tau^h)=\{\underline{\mu}, \lambda, \overline{\mu}\} = \{\frac{\lambda \varepsilon_1}{\overline{\tau}^d(\underline{\mu})}, \lambda, \frac{\lambda(1-\varepsilon_1)}{\overline{\tau}^d(\overline{\mu})}\}$ (these scalars denote the second entries of posterior beliefs) with $\Delta \theta = (\theta_1 - \theta_0) (\overline{\mu}-\underline{\mu})$, and 
    \begin{align*}
        &\tau^h(\underline{\mu})=\frac{\overline{\tau}^d(\underline{\mu})\nabla c(\lambda)}{D \Delta \theta},\quad \tau^h(\lambda)=1-\frac{\nabla c(\lambda)}{D \Delta \theta}, \\
        & \tau^h(\overline{\mu})=\frac{\overline{\tau}^d(\overline{\mu}) \nabla c(\lambda)}{D \Delta \theta}. 
    \end{align*}
    Note that $\overline{\tau}^d(\underline{\mu})=(1-\lambda)(1-\varepsilon_0)+\lambda\varepsilon_1$ and $\overline{\tau}^d(\overline{\mu})=(1-\lambda)\varepsilon_0+\lambda(1-\varepsilon_1)$ are for the distribution over posteriors under fully-informative tagging. Then, we can verify that the hybrid posterior distribution $\tau^h$ satisfies both constraints in the sender's problem. For the first constraint, 
    \begin{align*}
        \E_{\tau^h}[\mu] &= \frac{\overline{\tau}^d(\underline{\mu})\nabla c(\lambda)}{D \Delta \theta} \frac{\lambda \varepsilon_1}{\overline{\tau}^d(\underline{\mu})} \\
        & \quad + \left[1-\frac{\nabla c(\lambda)}{D \Delta \theta}\right] \lambda + \frac{\overline{\tau}^d(\overline{\mu}) \nabla c(\lambda)}{D \Delta \theta} \frac{\lambda (1-\varepsilon_1)}{\overline{\tau}^d(\overline{\mu})}=\lambda.
    \end{align*}
    As for the second constraint, 
    \begin{align*}
        &\E_{\tau^h} [f(\mu)]\\ &= \frac{\overline{\tau}^d(\underline{\mu})\nabla c(\lambda)}{D \Delta \theta} \left[\frac{-1}{1-\lambda} \frac{(1-\lambda)(1-\varepsilon_0)}{\overline{\tau}^d(\underline{\mu})} + \frac{1}{\lambda}\frac{\lambda\varepsilon_1}{\overline{\tau}^d(\underline{\mu})}\right]\frac{\lambda\varepsilon_1}{\overline{\tau}^d(\underline{\mu})} \\
        & \quad +\left(1-\frac{\nabla c(\lambda)}{D \Delta \theta}\right)\left[\frac{-1}{1-\lambda}(1-\lambda) + \frac{1}{\lambda} \lambda \right]\lambda -\nabla c(\lambda)\\
        & \quad + \frac{\overline{\tau}^d(\overline{\mu}) \nabla c(\lambda)}{D \Delta \theta} \left[\frac{-1}{1-\lambda} \frac{(1-\lambda) \varepsilon_0}{\overline{\tau}^d(\overline{\mu})} + \frac{1}{\lambda}\frac{\lambda (1-\varepsilon_1)}{\overline{\tau}^d(\overline{\mu})}\right]\frac{\lambda (1-\varepsilon_1)}{\overline{\tau}^d(\overline{\mu})}\\
        &= (1-\varepsilon_0-\varepsilon_1)\frac{\nabla c(\lambda)}{D \Delta \theta} (\overline{\mu}-\underline{\mu}) -\nabla c(\lambda)=0.
    \end{align*}
     For the special case $\lambda=0$, $\underline{\mu}=\overline{\mu}=0$, and $\operatorname{supp}(\tau^h)$ reduces to $\{0\}$ and $\tau^h(0)=1$, which also satisfies both constraints in the sender's problem. This construct implies that for any $\lambda\in [0, \bar{\lambda}]$, one can find a feasible $\tau$, and hence, $\lambda$ is also implementable.
\end{proof}

\begin{corollary}
    $\lambda=0$ is implementable under arbitrary signaling while $\bar{\lambda}$ is implementable if and only if the signaling is fully informative. 
\label{cor:lambda_bar}
\end{corollary}
\begin{proof}
    The proof is provided in Appendix \ref{app:lambda_bar}.
\end{proof}

The above discussion addresses the feasibility condition for the agent. We now turn to the sender's problem, given an implementable effort $\lambda$. Let $\tau^\lambda$ and $V^\lambda$ denote the optimal solution to the sender's problem \eqref{eq;sender-max} with fixed $\lambda$, and the corresponding objective value, respectively. Define the set $F^\lambda\subset \R^{|\Omega|+2}$: $F^\lambda=\{(\mu, f(\mu), v_S(\mu)): \mu\in [\underline{\mu}, \overline{\mu}]\}$. By construction, each entry of any element in $F^\lambda$ corresponds to the integrand in the three objects in the sender's problem \eqref{eq;sender-max}. These integrands are referred to as ex-post values. Let $co(F^\lambda)$ denote the convex hull of $F^\lambda$, which includes all the ex-ante values generated by a probability distribution $\overline{\tau}^d \in \Delta([\underline{\mu}, \overline{\mu}])$. The following proposition offers a geometric insight behind the Lagrangian multiplier method that is widely employed in single-agent constrained optimization \cite{andrzej06nonlinear}, multi-agent constrained games and generalized Nash equilibrium \cite{facchinei10generalized-NE, peng21net-constraint, shutian23erm}, and constrained reinforcement learning \cite{altman99cmadp, tao24dima}. 

\begin{proposition}
    Given an implementable effort $\lambda$, the maximal utility the sender can attain is $V^\lambda=\max\{v: (\theta(\lambda), 0, v)\in co(F^\lambda)\}$.
\label{prop:v_lambda}
\end{proposition}
\begin{proof}
    The proof is provided in Appendix \ref{app:V_lambda}.
\end{proof}

The proposition provides geometric insight into the solution: the point $(\theta(\lambda), 0, V^\lambda)$ lies on the boundary of the convex set $co(F^\lambda)$. Therefore, a supporting hyperplane exists at $(\theta(\lambda), 0, V^\lambda)$, which leads to the following characterization.

\begin{theorem}[Lagrangian Characterization]
    Given an implementable $\lambda$, a distribution of posteriors $\tau^\lambda$ is a solution to the sender's problem if and only if it satisfies \eqref{eq:bp}, \eqref{eq: ic}, and there exists $\psi\in \R$, $\rho\in \R$, and $\varphi\in \R^{|\Omega|}$ such that 
    \begin{equation*}
        \mathcal{L}(\mu, \psi,\varphi)=v_S(\mu)+\psi f(\mu)-\langle \varphi , \mu \rangle \leq \rho , \forall \mu\in \Delta^d(\Omega),\label{eq:lagrangian}
    \end{equation*}
    where the equality holds for all $\mu$ such that $\tau^\lambda(\mu)>0$.
\label{prop:Lagrangian}
\end{theorem}
\begin{proof}
    The proof is provided in Appendix \ref{app:Lagrangian}.
\end{proof}

Note that the introduced Lagrangian function $\mathcal{L}(\mu, \psi,\varphi)$ is concerned with the ex-post values (i.e., the belief is realized) while the sender's problem is of ex-ante. Hence, one should inspect the expectation of such Lagrangian with respect to the posterior distribution $\tau^\lambda$: $\mathbb{E}_{\tau^\lambda}[L(\mu, \psi,\varphi)]$. In this case, the convexity/concavity of $\mathcal{L}(\mu, \psi,\varphi)$ becomes a pivotal issue. 

Fixing $\lambda\in (0, \bar{\lambda})$, the Lagrangian's second-order derivative is given by $\frac{\partial^2 \mathcal{L}}{\partial \mu^2} = \nabla^2 v_S(\mu) +\frac{2\psi}{\lambda(1-\lambda)}$.  The sign of $\frac{\partial^2 \mathcal{L}}{\partial \mu^2}$, which indicates the Lagrangian's convexity, is determined by the signs of $\nabla^2 v_S(\mu)$ and $\psi$. Our prior work continues the characterization by inspecting the sign of $\psi$ and the convexity of the Lagrangian in\cite[Prop. 7]{yang2023designing}, which we briefly revisit below. 
\begin{proposition}
    For any $\lambda\in (0,\bar{\lambda}]$, the Lagrange multiplier $\psi$ associated with the solution $\tau^\lambda$ is non-positive. 
\end{proposition}
\begin{proof}
Consider a relaxation to the original problem without IC constraint \eqref{eq: ic}:
\begin{align}
    \widetilde{V}^\lambda=\max_{\tau} \E_{\tau}[v_S(\mu)] \text{ subject to } \E_\tau[\mu]=p(\lambda),
\label{eq:relax}
\end{align}
which is exactly the standard Bayesian persuasion \cite{kamenica11BP}.

Denote by $\tilde{\tau}^\lambda$ the solution to the relaxed problem when fixing $\lambda$. Applying the Lagrangian characterization developed in Theorem \ref{prop:Lagrangian}, there exists $\tilde{\rho}$ and $\tilde{\varphi}$ such that $v_S(\mu)\leq \tilde{\rho} + \tilde{\varphi} \mu$, for all $\mu\in [0,1]$, with equality if $\tilde{\tau}(\mu)>0$.  Define $g(\lambda)=\E_{\tilde{\tau}}[f(\mu)]$. Let $\tau^\lambda$ be the solution to the original problem. We aim to prove $\psi g(\lambda)\leq 0$ in the following. The definition of two Lagrangians give
\begin{align}
    \rho+\varphi \lambda=\E_{\tau^\lambda}[v_S(\mu)]\leq \E_{\tilde{\tau}^\lambda}[v_S(\mu)]=\tilde{\rho}+\tilde{\varphi}\lambda.\label{eq:tilde-ineq}
\end{align}
Finally, taking the expectation of the original Lagrangian in \Cref{prop:Lagrangian} with respect to $\tilde{\tau}$, we obtain 
\begin{equation}
    \E_{\tilde{\tau}}[v_S(\mu)]+\psi \E_{\tilde{\tau}}[f(\mu)]\leq \rho +\varphi \lambda \Leftrightarrow  \tilde{\rho} +\tilde{\varphi} \lambda +\psi g(\lambda) \leq \rho +\varphi \lambda \label{eq:expect-ineq}
\end{equation}
Combining \eqref{eq:expect-ineq} and \eqref{eq:tilde-ineq} leads to $\psi g(\lambda)\leq 0$. 

The rest of the proof establishes that $g(\lambda)\geq0$ for $\lambda \in (0,\bar{\lambda}]$. Note that the sender's expected utility $v_S(\mu)$ is convex in $\mu$. The standard persuasion analysis gives that the unique optimal signaling is the fully informative one \cite[Section 3]{kamenica11BP}, implying that $\operatorname{supp}(\tilde{\tau})=\{\underline{\mu}, \overline{\mu}\}$, and $\tilde{\tau}(\underline{\mu})=\overline{\tau}^d(\underline{\mu})$, $\tilde{\tau}(\overline{\mu})=\overline{\tau}^d(\overline{\mu})$. Direct calculation yields $g(\lambda)=\E_{\tilde{\tau}}[f(\mu)]=D(\overline{\mu}-\underline{\mu})-\nabla c(\lambda)\geq 0$ according to $(\theta_1-\theta_0)D(\overline{\mu}-\underline{\mu}) \geq \nabla c(\lambda)$ in the proof of Proposition \ref{prop:feasible}. Hence, for $\lambda\in (0,\bar{\lambda}]$, $g(\lambda) \geq 0$ implies that $\psi\leq 0$. 
\end{proof}

Although $\frac{\partial^2 \mathcal{L}}{\partial \mu^2}$ is not necessarily non-negative, our prior work \cite[Proposition 8]{yang2023designing} asserts that the Lagrangian must be a convex function of $\mu$. This is established by contradiction: if $\frac{\partial^2 \mathcal{L}}{\partial \mu^2} < 0$, the sender's optimal signaling is degenerate (only one belief) and is strictly dominated by the hybrid signaling, which contradicts the optimality. However, the posterior belief space shrinks under misdetection, as considered in this work, and the resulting hybrid signaling does not necessarily retain strict dominance. In such cases, the assumption that the degenerate signaling is strictly dominated may no longer hold.

In this work, we provide an alternative perspective to show that even if there exists a misperception (or detection error) $d$, the sender’s optimal
signaling is still the fully informative one, under which the agent
is incentivized not to create misinformation to the best effort.

\begin{theorem}
Given the sender's utility $v_S(\mu)$ is convex and non-decreasing, the optimal signaling is fully informative and encourages the agent to implement $\bar{\lambda}$.
\label{prop:optimal-is}
\end{theorem}

\begin{proof}
    Consider the relaxation in \eqref{eq:relax},
which is exactly the standard Bayesian persuasion \cite{kamenica11BP}. 
Due to $v_S(\mu)$ being convex, the standard analysis gives that $\widetilde{V}^\lambda$ is attained by the fully-informative signaling \cite[Section 3]{kamenica11BP}.
Then, under fully-informative $\overline{\tau}^d$ with $\supp=\{\underline{\mu}, \overline{\mu}\}$ in Section \ref{sec:feasible_mu}, we have
\begin{align*}
    \widetilde{V}^\lambda &= \E_{\overline{\tau}^d}[v_S(\mu)]=\overline{\tau}^d(\underline{\mu})v_S(\underline{\mu}) + \overline{\tau}^d(\overline{\mu})v_S(\overline{\mu})\\
    \Rightarrow \nabla_{\lambda} \widetilde{V}^\lambda &= \nabla_{\lambda} \overline{\tau}^d(\underline{\mu}(\lambda)) v_S(\underline{\mu})+\overline{\tau}^d(\underline{\mu})\nabla_{\lambda} v_S(\underline{\mu}(\lambda))\\
    & \quad + \nabla_{\lambda} \overline{\tau}^d(\overline{\mu}(\lambda)) v_S(\overline{\mu})+\overline{\tau}^d(\overline{\mu})\nabla_{\lambda} v_S(\overline{\mu}(\lambda))\\
    &=(1-\varepsilon_0-\varepsilon_1)\left(v_S(\overline{\mu})-v_S(\underline{\mu})\right)\\
    &\quad + \overline{\tau}^d(\underline{\mu}) \nabla_{\underline{\mu}}v_S(\underline{\mu})\nabla_{\lambda}\underline{\mu}(\lambda)\\ 
    &\quad + \overline{\tau}^d(\overline{\mu}) \nabla_{\overline{\mu}}v_S(\overline{\mu})\nabla_{\lambda}\overline{\mu}(\lambda),
\end{align*}which is positive as $(1-\varepsilon_0-\varepsilon_1)>0$, $v_S(\mu)$ is non-decreasing, and $\underline{\mu}, \overline{\mu}$ are both increasing in $\lambda$. This implies that $\widetilde{V}^{\bar{\lambda}} \geq \widetilde{V}^{\lambda}$ for $\lambda \in [0, \bar{\lambda}]$.
Note that we can obtain: (i) $\widetilde{V}^\lambda \geq V^\lambda$, as there is no IC constraint for the relaxed problem;  (ii) $\widetilde{V}^{\bar{\lambda}} = V^{\bar{\lambda}}$ since $\bar{\lambda}$ is only implementable under fully-informative tagging in the original problem (see Corollary \ref{cor:lambda_bar}), making the objective values for the relaxed and the original problems equivalent. Hence, we have $V^{\bar{\lambda}} = \widetilde{V}^{\bar{\lambda}} \geq \widetilde{V}^\lambda \geq V^\lambda$, which indicates that the optimal signaling for the original problem is fully-informative and encourages $\bar{\lambda}$.

\end{proof}

\vspace{-0.5cm}
\section{Numerical Studies}\label{sec:numerical}
This section first studies the proposed Bayesian persuaded branching processes model under the fully informative tagging policy, and then compares the fully informative tagging with hybrid tagging in the proof of Proposition \ref{prop:feasible}. For each experiment, the branching setup is given by  $N_0=P_0=50, m_{N}=50, q=0.5$, and $t_{i}, i=1, \cdots, 500$. The numerical results in this section are the average of $200$ independent simulations.
We assume the content provider's effort cost is given by $c(\lambda)=k\lambda^2$, with cost coefficient $k > 0.5$, satisfying Assumption \ref{assump:cost}. 

\subsection{Fully Informative Tagging Policy}
Under the fully informative tagging policy, the SNP tags the post according to the perceived true state, i.e., $\sigma= \omega^\prime$, which incentivizes the content provider to exert maximum effort $\bar{\lambda}$ (according to Proposition \ref{cor:lambda_bar}), leading to the following posterior beliefs: $$\underline{\mu}=\frac{\bar{\lambda} \varepsilon_1}{(1-\bar{\lambda})(1-\varepsilon_0)+\bar{\lambda}\varepsilon_1}, \overline{\mu}=\frac{\bar{\lambda}(1-\varepsilon_1)}{(1-\bar{\lambda})\varepsilon_0+\bar{\lambda}(1-\varepsilon_1)}.$$ From Proposition \ref{prop:feasible}, $\bar{\lambda}$ is determined by $\nabla c(\bar{\lambda})=(\theta_1-\theta_0)D(\overline{\mu}-\underline{\mu})$, where $\nabla c(\bar{\lambda})=2k\bar{\lambda}$ and $D=1-\varepsilon_0-\varepsilon_1$. By solving for $\bar{\lambda}$, we arrive at the case setup outlined in Table \ref{tab:setup1}, where we explore scenarios with varying cost coefficients for the content provider and different false alarm rates for the SNP’s detection errors. For instance, when cost coefficient $k=0.6$, the maximum implementable effort is $\bar{\lambda}=0.66$ when false alarms are $\varepsilon_0=\varepsilon_1=0.05$, while $\bar{\lambda}=0.34$ when false alarms are $\varepsilon_0=0.15$ and $=\varepsilon_1=0.20$.

\begin{figure}[!ht]
    \centering
    \begin{subfigure}[t]{0.49\textwidth}
        \centering
        \includegraphics[width=3in]{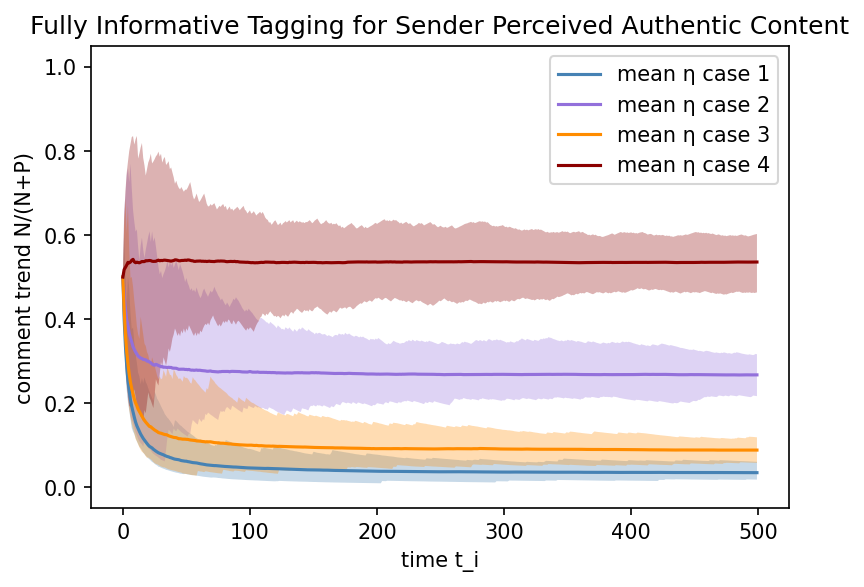}
        \caption{The SNP's perceived authentic post yields a rather positive trend under the fully informative policy.}
        \label{fig:fully_real}
    \end{subfigure}
    \hfill
    \begin{subfigure}[t]{0.49\textwidth}
        \centering
        \includegraphics[width=3.1in]{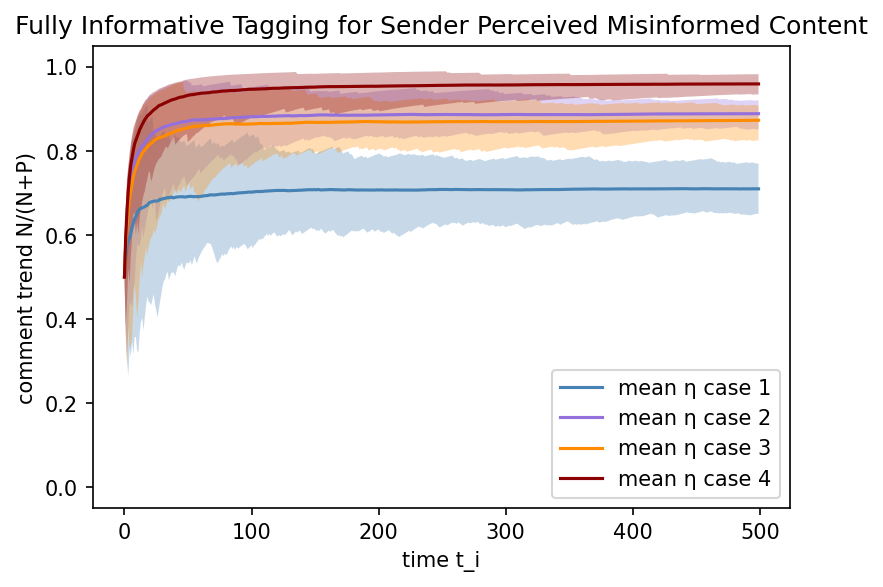}
        \caption{The SNP's perceived post with misinformation yields a negative trend under the fully informative policy.}
        \label{fig:fully_fake}
    \end{subfigure}
    \caption{Simulations of online misinformation circulation under fully informative tagging policy. The shaded region indicates the standard deviation of $\eta^*$, while the line represents the mean of $\eta^*$.}
\label{fig:fully}
\vspace{-3mm}
\end{figure}

By comparing $\bar{\lambda}$ between cases 1 and 3 in Table \ref{tab:setup1}, we observe that a larger cost coefficient $k$ results in a lower maximum implementable effort $\bar{\lambda}$, which is intuitive since higher costs for unveiling the truth deter the content provider from investing effort. Similar results can be obtained between cases 2 and 4. Additionally, higher false alarm rates $\varepsilon_0$ and $\varepsilon_1$ in misinformation detection also reduce the maximum implementable effort $\bar{\lambda}$, as can be seen in cases 1 and 2. This is because, from the content provider's perspective, even with significant effort invested in creating authentic content, the post may still be mis-tagged due to detection errors, leading to negative trends of the post.

\begin{table}[htbp]
\caption{Case Study Under Fully Informative Tagging}
\begin{center}
\begin{tabular}{ccccc}
\toprule
Case \# & cost coefficient $k$ & $\varepsilon_0$ & $\varepsilon_1$ & maximum effort $\bar{\lambda}$ \\
\midrule
$1$ & $0.6$ & $0.05$  & $0.05$ & $0.66$ \\
$2$ & $0.6$ & $0.15$  & $0.20$ & $0.34$  \\
$3$ & $1.0$ & $0.05$  & $0.05$ & $0.40$ \\
$4$ & $1.0$ & $0.15$  & $0.20$ & $0.14$  \\
\bottomrule \\[-0.3em]
\multicolumn{5}{c}{Effort cost $c(\lambda)=k\lambda^2$.}
\end{tabular}
\end{center}
\label{tab:setup1}
\vspace{-5mm}
\end{table}

For the perceived authentic post, $\omega^\prime=1, \sigma=1, \E_{\mu_\sigma}[\omega]=\overline{\mu}$,
and thus $\alpha_{yx}=\alpha_{xx}= 1-\E_{\mu_\sigma}[\omega]=1-\overline{\mu}$. The results for the proportion of negative comments $\eta^*$ are shown in \Cref{fig:fully_real}, demonstrating that higher maximum implementable effort leads to more positive trends. This suggests that reducing detection errors from the platform is crucial for encouraging content providers to invest more effort, which drives positive trends on the platform. On the other hand, for the perceived misinformed post, $\omega^\prime=0, \sigma=0, \E_{\mu_\sigma}[\omega]=\underline{\mu}$, we similarly have $\alpha_{yx}=\alpha_{xx}= 1-\E_{\mu_\sigma}[\omega]=1-\underline{\mu}$. The results for the proportion of negative comments $\eta^*$ in \Cref{fig:fully_fake} reveal that tag $\sigma=0$ yields a negative trend, with lower maximum implementable effort leading to more negative trends.

\subsection{Hybrid Informative Tagging Policy}
Under the hybrid tagging policy specified in \Cref{prop:feasible}, any $\lambda \in (0, \bar{\lambda})$ is implementable; we then consider the case setup listed in Table \ref{tab:setup2}.

\begin{table}[htbp]
\caption{Case Study Under Hybrid Informative Tagging}
\begin{center}
\begin{tabular}{cccccc}
\toprule
Case \# & coefficient $k$ & $\varepsilon_0$ & $\varepsilon_1$ & maximum $\bar{\lambda}$ & chosen $\lambda$ \\
\midrule
$1.1$ & $0.6$ & $0.05$  & $0.05$ & $0.66$ & $0.65$\\
$1.2$ & $0.6$ & $0.05$  & $0.05$ & $0.66$ & $0.40$\\
$2.1$ & $0.6$ & $0.15$  & $0.20$ & $0.34$ & $0.33$\\
$2.2$ & $0.6$ & $0.15$  & $0.20$ & $0.34$ & $0.10$ \\
$3.1$ & $1.0$ & $0.05$  & $0.05$ & $0.40$ & $0.39$\\
$3.2$ & $1.0$ & $0.05$  & $0.05$ & $0.40$ & $0.20$ \\
$4.1$ & $1.0$ & $0.15$  & $0.20$ & $0.14$ & $0.13$\\
$4.2$ & $1.0$ & $0.15$  & $0.20$ & $0.14$ & $0.07$ \\
\bottomrule \\[-0.3em]
\multicolumn{6}{c}{Effort cost $c(\lambda)=k\lambda^2$. chosen $\lambda \leq \bar{\lambda}$.}
\end{tabular}
\end{center}
\label{tab:setup2}
\vspace{-5mm}
\end{table}

For $k=0.6$, the maximum effort is $\bar{\lambda}=0.66$ when the false alarm rates are $\varepsilon_0 = \varepsilon_1 = 0.05$. In this scenario, if the content provider chooses to invest effort $\lambda = 0.65$, the post's comment trend is more positive compared to an investment of $\lambda = 0.4$, as illustrated in case 1.1 and 1.2 of Fig. \ref{fig:hybrid_1_2}. Hence, the more effort the content provider spends, the more positive the post's trend is, and the higher the reputation the content provider can earn. 

However, as the false alarm rates increase, the range of implementable efforts narrows, and the resulting $\eta^*$ becomes strictly greater than half, as demonstrated in cases 2.1 and 2.2 of Fig. \ref{fig:hybrid_1_2}. In this scenario, the content provider may opt not to generate content, as their effort leads to negative trends in the mean under hybrid tagging. 
Similar consequences happen when the cost (cost coefficient) for investigating the truth is too high, as shown in Fig. \ref{fig:hybrid_3_4}. From the SNP's perspective, discouraging UGC generation is not rational if the goal is to maintain an active and engaging platform. Therefore, the SNP opts for a fully informative tagging policy and reducing detection errors to achieve trend outcomes similar to cases 1 and 3 in Fig. \ref{fig:fully}. 

\begin{figure}[!ht]
    \centering
    \begin{subfigure}[t]{0.49\textwidth}
        \centering
        \includegraphics[width=2.5in]{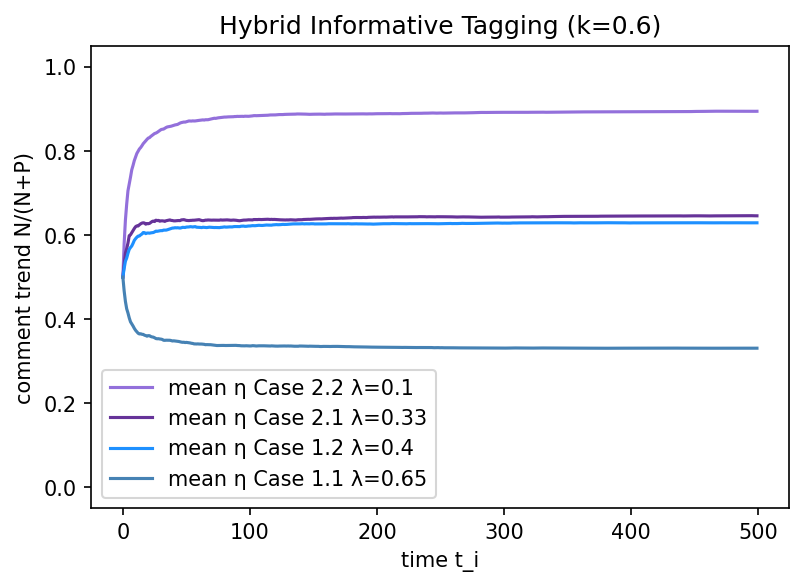}
        \caption{Hybrid tagging policy when the cost coefficient is $k=0.6$ for the content provider.}
        \label{fig:hybrid_1_2}
    \end{subfigure}
    \hfill
    \begin{subfigure}[t]{0.49\textwidth}
        \centering
        \includegraphics[width=2.6in]{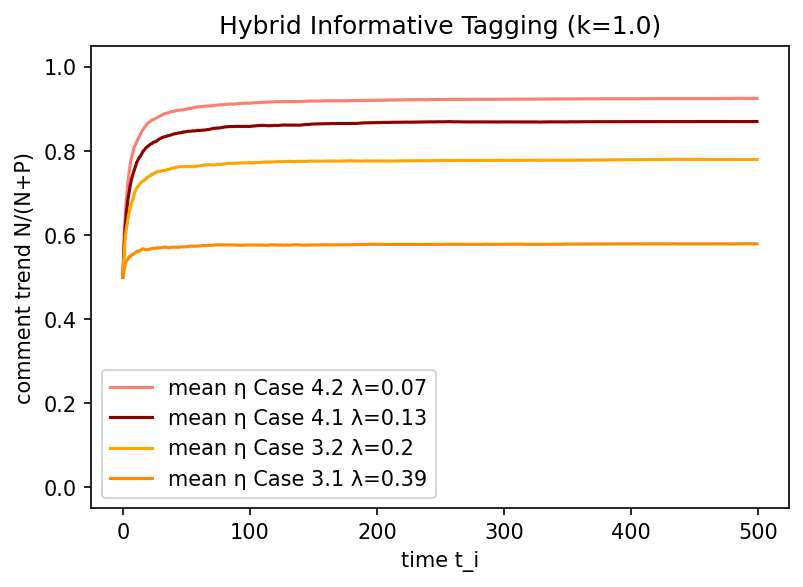}
        \caption{Hybrid tagging policy when the cost coefficient is $k=1.0$ for the content provider.}
        \label{fig:hybrid_3_4}
    \end{subfigure}
    \caption{Simulations of online misinformation circulation when the SNP adopts a hybrid tagging policy. The line represents the mean of $\eta^*$.}
\label{fig:hybrid}
\vspace{-3mm}
\end{figure}

\section{Conclusion}
This work has investigated a preemptive approach to mitigate misinformation spread on SNP by incentivizing the content provider to generate authentic content in the first place. When designing tagging policies to leverage social nudges from population-level user responses, SNP must be cautious about
the potential detection errors of misinformation. Hence, we have developed a three-player persuasion game to model the strategic interaction under misdetection among the SNP, the content provider, and the user, with the spread of misinformation content modeled as a multi-type branching process. By transforming the perfect Bayesian equilibrium into the posterior belief space influenced by detection errors, we have reformulated the SNP's equilibrium as an equality-constrained convex optimization problem, which admits a concise Lagrangian characterization. We show that the SNP’s optimal policy is still transparent tagging, i.e., revealing the content’s perceived authenticity, to the user despite detection errors, which nudges the provider not to generate misinformation, even though the SNP exerts no direct control over the UGC from the content provider.  
One direction of future work would be to explore cases where misdetection is unknown to the content provider, users, or both. The SNP might choose not to disclose false alarms to protect its reputation or sustain user engagement on the platform.

\bibliography{reference}
\bibliographystyle{IEEEtran}

\appendices
\renewcommand{\thesectiondis}[2]{\Alph{section}:}
\section{Proof of Proposition \ref{prop:increasing_mu}} 
\label{app:increasing_mu}
\begin{proof}
    We can show that the partial derivatives
    \begin{align*}
        &\frac{\partial \underline{\mu}}{\partial \lambda}=\frac{\varepsilon_1(1-\varepsilon_0)}{((1-\lambda)(1-\varepsilon_0)+\lambda\varepsilon_1)^2} \geq 0,\\
        &\frac{\partial^2 \underline{\mu}}{\partial \lambda^2}=\frac{2\varepsilon_1(1-\varepsilon_0)(1-\varepsilon_0-\varepsilon_1)}{((1-\lambda)(1-\varepsilon_0)+\lambda\varepsilon_1)^3} \geq 0, \\
        &\frac{\partial \overline{\mu}}{\partial \lambda}=\frac{\varepsilon_0(1-\varepsilon_1)}{((1-\lambda)\varepsilon_0+\lambda(1-\varepsilon_1))^2}\geq 0,\\
        &\frac{\partial^2 \overline{\mu}}{\partial \lambda^2}=\frac{-2\varepsilon_0(1-\varepsilon_1)(1-\varepsilon_0-\varepsilon_1)}{((1-\lambda)\varepsilon_0+\lambda(1-\varepsilon_1))^3}\leq 0,
    \end{align*} as $\varepsilon_0+\varepsilon_1 < 1$, which then complete the proof.
\end{proof} 
\vspace{-1em}
\section{Proof of Corollary \ref{cor:lambda_bar}}\label{app:lambda_bar}
\begin{proof}
    When $\lambda=0$, the prior becomes $\theta_0$. Hence, $\mu_\sigma=0$ regardless of the signaling mechanism $\pi$. Note that $v_A(\mu)=\mu$, and therefore $\bar{v}_A^d(\omega|\pi):=\sum_{\sigma}\sum_{\omega'}d(\omega'|\omega)\pi(\sigma|\omega')v_A(\mu_\sigma)=0$ for any $\omega$ when $\lambda =0$. Consequently, when $\lambda=0$, \eqref{eq:agent-opt} holds for arbitrary $\pi$, as $\nabla c(\lambda)=0$. When $\lambda=\bar{\lambda}$, we begin with the necessity.
    Recall from Proposition \ref{prop:feasible} that the IC constraint \eqref{eq:agent-opt} is 
    \begin{align*}
        \nabla c(\lambda) &= (\theta_1-\theta_0)D(\pi(1|1)+\pi(0|0)-1)(\mu_1-\mu_0). 
    \end{align*} When $\lambda=\bar{\lambda}$, $\nabla c(\lambda)= (\theta_1-\theta_0)D(\overline{\mu}-\underline{\mu})$, and \eqref{eq:agent-opt} holds (i.e., $\bar{\lambda}$ is implementable) if $\pi(1|1)=1, \pi(0|0)=1$, $\mu_1 = \overline{\mu}, \mu_0=\underline{\mu}$, which is the fully-informative signaling case. For sufficiency, as $\bar{\lambda}$ directly satisfies the IC constraint \eqref{eq:agent-opt} when the signaling is fully-informative, $\bar{\lambda}$ is also implementable.
\end{proof}
\vspace{-2.5em}
\section{Proof of Proposition \ref{prop:v_lambda}}
\label{app:V_lambda}
\begin{proof}
    It suffices to note that $\mu=\theta(\lambda)$ naturally satisfies \eqref{eq:bp}, and $f(\mu)=0$ induces \eqref{eq: ic}. Therefore, any point $(\mu, f(\mu), v)\in \{(\theta(\lambda), 0, v)\in co(F^\lambda)\}$ is feasible for \eqref{eq;sender-max}. Therefore, $V^\lambda$, being a convex combination of such points, represents the maximal value.
\end{proof}

\vspace{-1.5em}
\section{Proof of Theorem \ref{prop:Lagrangian}} 
\label{app:Lagrangian}
\begin{proof}
We begin with the necessity. As $(\theta(\lambda),0, V^\lambda )$ is a boundary point of a closed convex set, the separating hyperplane theorem tells that there exists a normal vector $b=(-\varphi, \psi, 1)\in \R^{|\Omega|+2}$ and a scalar $\rho$ such that $\langle b, x\rangle \leq \rho$ for all $x\in co(F^\lambda)$, where the equality holds for $x=(\theta(\lambda), 0, V^\lambda)$. Rearranging terms in this inner product, we obtain that $ \mathcal{L}(\mu, \psi,\varphi)\leq \rho $. 

It remains to show that $\mathcal{L}(\mu, \psi,\varphi)=\rho$ for all $\mu\in \{\mu:\tau^\lambda(\mu)>0\}$. Suppose, for the sake of contradiction, that there exists some $\mu\in \operatorname{supp}(\tau^\lambda)$ such that  $\mathcal{L}(\mu, \psi, \varphi)<\rho$. Note that $\mathcal{L}(\mu, \psi, \varphi)\leq\rho$, then $V^\lambda=\E_{\tau^\lambda}[\mathcal{L}(\mu, \psi, \varphi)]< \rho$. Rearranging terms, we obtain $\langle b, (\theta(\lambda), 0, V^\lambda)\rangle < \rho$, which contradicts the fact that the supporting hyperplane passes through the point $(\theta(\lambda), 0, V^\lambda)$.

For the sufficiency part, if $v_S(\mu)+\psi f(\mu)\leq \rho + \langle \varphi, \mu\rangle $ for all $\mu \in [\underline{\mu}, \overline{\mu}]$, then for any $\tau^d$, 
\begin{equation*}
    \E_{\tau^d }[v_S(\mu)]+\psi \E_{\tau^d}[f(\mu)]\leq \rho+\E_{\tau^d}[\langle \varphi, \mu\rangle].
\end{equation*}
Since $\tau^\lambda$ satisfies \eqref{eq:bp} and \eqref{eq: ic}, the above reduces to $\E_{\tau^\lambda }[v_S(\mu)]\leq \rho+ \langle \varphi, \theta(\lambda)\rangle$. If $\tau^\lambda$ is such that $\mathcal{L}(\mu,\psi, \varphi)=\rho$, for all $\mu\in\operatorname{supp}(\tau^\lambda)$, then $\E_{\tau^\lambda}[v_S(\mu)]=\rho+\langle \varphi, \theta(\lambda)\rangle $, meaning that the expected utility $\E_\tau[v_S(\mu)]$ reaches the upper bound at $\tau^\lambda$. 
\end{proof}

\end{document}